\documentclass{article}
\usepackage{macros} 
\DeclareMathOperator{\agr}{agr}

\title{Near-Optimal List-Recovery of Linear Code Families}

\author{Ray Li\thanks{Math \& CS Department, Santa Clara University. Email: \url{rli6@scu.edu}. Research supported by NSF grant CCF-2347371.} , Nikhil Shagrithaya\thanks{Computer Science and Engineering Department, University of Michigan, Ann Arbor. Email: \url{nshagri@umich.edu}. Research supported in part by NSF grants CCF-2236931 and CCF-2107345.}}
\date{}

\begin{document}    

\maketitle

\begin{abstract}

We prove several results on linear codes achieving list-recovery capacity.
We show that random linear codes achieve list-recovery capacity with constant output list size (independent of the alphabet size and length).
That is, over alphabets of size at least $\ell^{\Omega(1/\eps)}$, random linear codes of rate $R$ are $(1-R-\eps, \ell, (\ell/\varepsilon)^{O(\ell/\eps)})$-list-recoverable for all $R\in(0,1)$ and $\ell$.
Together with a result of Levi, Mosheiff, and Shagrithaya, this implies that randomly punctured Reed--Solomon codes also achieve list-recovery capacity. 
We also prove that our output list size is near-optimal among \emph{all} linear codes: all $(1-R-\eps, \ell, L)$-list-recoverable linear codes must have $L\ge \ell^{\Omega(R/\eps)}$.

Our simple upper bound combines the Zyablov-Pinsker argument with recent bounds from Kopparty, Ron-Zewi, Saraf, Wootters, and Tamo on the maximum intersection of a ``list-recovery ball" and a low-dimensional subspace with large distance.
Our lower bound is inspired by a recent lower bound of Chen and Zhang.
\end{abstract}

\section{Introduction}
In this work, we study list-recovery for random linear codes and random Reed--Solomon codes, proving near-optimal upper and lower bounds. 

An \deffont{(error correcting) code} $\calC$ is a subset of $\Sigma^n$ for an alphabet $\Sigma$, which, in this work, is always $\F_q$ for some prime power $q$.
We study \emph{linear} codes, which are subspaces of $\F_q^n$. 
We want codes to be large, meaning they have large \emph{rate} $R=\inparen{\log_q |\calC|}/n$.
We also want codes to tolerate more errors.
In the standard unique decoding setting, tolerating many errors means that, for any vector $z\in\F_q^n$, there is at most one \emph{codeword} $c\in\calC$ that agrees with $z$ on many coordinates.

We study a generalization of the unique-decoding problem known as list-recovery.
In list-recovery, we want that, for any $\ell\times\ell\times\cdots\times \ell$ combinatorial rectangle, there are few codewords $c$ that ``agree'' with this rectangle on many coordinates.
Formally, a code $\calC$ is \deffont{$\inparen{\rho, \ell, L}$-list-recoverable} if for any sets $S_1,\dots,S_n\subset \mathbb{F}_q$ of size $|S_i|=\ell$, there are at most $L$ codewords $c_1,\ldots, c_L \in \calC$ such that $c_i \in  S_i$ for at least a $(1-\rho) n$ fraction of the coordinates. The special case of $(\rho, 1,1)$ list-recoverability is the standard unique-decoding setting, and the special case of $(\rho, 1, L)$ list-recovery is known as \deffont{list-decoding}.

List-recovery has motivations in coding theory, complexity theory, and algorithms.
In coding theory, list-recovery has been used as a tool to obtain efficient list-decoding algorithms \cite{guruswami1999improved,guruswami2008explicit,guruswami2013linearalgebraic,kopparty2015listdecoding,hemenway2019local}.
Also, list-recoverable random-linear codes --- which we study in this work --- are used as a building block in other coding constructions \cite{guruswami2001expanderbased,hemenway2018lineartime}.
In complexity theory, list-recoverable codes find applications in constructions of other pseudorandom objects such as extractors \cite{TZ04} and condensers \cite{GUV09}.
In algorithms, they are also useful primitives in group testing \cite{indyk2010efficiently,ngo2011efficiently} sparse recovery \cite{gilbert2013ell2}, and streaming algorithms \cite{larsen2019heavy,doron2020highprobability}. 

The list-recovery capacity theorem states that $\rho=1-R$ is the optimal tradeoff between the error radius $\rho$ and the code rate $R$. (see e.g., \cite{guruswami2019essential, resch2020list}).
That is, below capacity $\rho<1-R$, there exist $(p,\ell,O_\ell(1))$-list-recoverable codes of rate $R$, and above capacity $\rho>1-R$, any $(\rho,\ell,L)$-list-recoverable code must have exponential list size $L\ge q^{\Omega(n)}$.
The existence holds because uniformly random codes of rate $R$ (over sufficiently large alphabets $q\ge \ell^{\Omega(1/\varepsilon)}$) are $(1-R-\varepsilon,\ell,O(\ell/\varepsilon))$-list-recoverable with high probability.

We wish to understand what kinds of codes achieve list-recovery capacity.
A number of explicit code constructions are known to achieve list-recovery capacity, including Folded Reed--Solomon codes, Multiplicity codes, Folded Algebraic--Geometry codes. Additional techniques --- subspace evasive sets, subspace designs, and expander techniques --- can be used to improve the output list-size $L$ and alphabet size $q$ \cite{guruswami2008explicit,kopparty2015listdecoding,guruswami2013linearalgebraic,guruswami2012folded,guruswami2013list,hemenway2018lineartime,hemenway2019local,guruswami2016explicit,dvir2012subspace,kopparty2018improved,tamo2023tighter} (see Table 1 in \cite{kopparty2018improved}, see also \cite{srivastava2024improved, chen2024explicit} for even tighter list size bounds in the special case of list-decoding).

Still, several fundamental questions remain open. 
\begin{enumerate}
    \item  First, how list-recoverable is a random linear code? A random linear code is a random subspace of $\mathbb{F}_q^n$. All explicit constructions are based on linear codes (though many are only linear over a subfield), so it is natural to wonder about list-recovery of a ``typical'' linear code. As list-recovery is a pseudorandom property, this question also addresses the deeper geometric question of ``how similar is a random subspace to a random set over $\mathbb{F}_q^n$?'', which is well-studied in the more specific context of list-decoding \cite{zyablov1981list, elias1991errorcorrecting,guruswami2011listdecodability, CGV13, wootters2013list, rudra2014every,RW15, rudra2018averageradius,li2020improved,guruswami2021bounds,AGL24}.
\item Second, how list-recoverable are Reed--Solomon codes? The above constructions all generalize the Reed--Solomon code, the most fundamental polynomial evaluation code. Can Reed--Solomon codes themselves achieve list-recovery capacity? Given recent progress that showed the special case that Reed--Solomon codes achieve list-\emph{decoding} capacity \cite{brakensiek2023generic, guo2023randomly, AGL24}, this general case of list-recovery has been an obvious and tantalizing open question.
\item Lastly, is there a fundamental separation between linear and nonlinear codes for list-recovery? On one hand, there is no apparent separation for the special case of list-decoding, where random linear codes are list-decodable to capacity with list-size $O(1/\varepsilon)$ \cite{guruswami2011listdecodability,wootters2013list,li2020improved,guruswami2021bounds}, just like uniformly random codes. On the other hand, uniformly random codes are list-recoverable with list size $O(\ell/\varepsilon)$, but all known linear constructions require output list size at least $\ell^{\Omega(1/\varepsilon)}$, and this lower bound has been proven in various specific settings \cite{guruswami2021bounds,levi2024random, chen2024explicit}.
\end{enumerate}

We answer all three questions.
We show that random linear codes are list-recoverable to capacity with provably near-optimal output list size.
By a recent result of \cite{levi2024random}, this implies that randomly punctured Reed--Solomon codes are list-recoverable to capacity with near-optimal output list size.
Lastly, we prove a fundamental separation between linear and non-linear codes by showing that \emph{all} linear codes of rate $R$ must have list-size at least $L\ge\ell^{\Omega(R/\varepsilon)}$.

\subsection{Our results}
We now state our results in the context of prior work.

\begin{table}[]
    \centering
    \begin{tabular}{c|c|c|c}
        Citation & Radius $\rho$ & input list size &output list size \\[0.5ex] 
        \hline & & &\\[-1.8ex]
        \cite{zyablov1981list,Gur04} & $1-R-\varepsilon$ & $\ell$ & $q^{O(\ell/\varepsilon)}$ \\
        \cite{rudra2018averageradius} & $1-R-\varepsilon$ & $\ell$ & $q^{O(\log^2(\ell/\varepsilon))}$ \\
        This work & $1-R-\varepsilon$ & $\ell$ & $\ell^{O(\ell/\varepsilon)}$ \\
    \end{tabular}
    \caption{List-recovery of Random Linear codes}
    \label{tab:rlc}
\end{table}

\paragraph{List recovery for Random Linear Codes.} 
Several known arguments show that random linear codes achieve list-recovery capacity.
A random linear code is a code generated by a uniformly random generator matrix $\mathbf{G} \in\F_q^{n\times k}$.
First, the Zyablov-Pinsker argument \cite{zyablov1981list} adapted to list-recovery shows that random linear codes of rate $R$ over alphabet $q\ge \ell^{\Omega(1/\varepsilon)}$ are $(1-R-\varepsilon, \ell, q^{O(\ell/\varepsilon)})$-list-recoverable (see, for example \cite[Lemma 9.6]{Gur04}).
Rudra and Wootters \cite{rudra2018averageradius} improved the output list size to $q^{O(\log^2(\ell/\varepsilon))}$, showing random linear codes of rate $R$ over alphabet $q\ge \ell^{\Omega(1/\varepsilon)}$ are $(1-R-\varepsilon, \ell, q^{O(\log^2(\ell/\varepsilon))})$-list-recoverable. 
We improve the output list size to be independent of the alphabet size $q$.
\begin{theorem}[Theorem~\ref{thm:list-rec-rlcs}, Informal]
    For all $R,\varepsilon\in(0,1)$, and $q\ge \ell^{\Omega(1/\varepsilon)}$ a random linear code of rate $R$ is $\inparen{1-R-\eps, \ell, \inparen{\frac{\ell}{\eps}}^{O(\ell/\eps)}}$-list-recoverable with high probability.
    \label{thm:main-1}
\end{theorem}
Our list size improves on the prior bounds when $q\ge \ell^{\Omega(\ell/\varepsilon)}$, which covers most alphabet sizes ($q\ge \ell^{\Omega(1/\varepsilon)}$ is needed to achieve list-recovery capacity), and the improvement is more significant when $q$ is larger.
This improvement to an alphabet-independent output list size is critical for Theorem~\ref{thm:main-2} below (see Remark~\ref{rem:main-2}).
As we show in Theorem~\ref{thm:lb}, this output list size is near optimal among all linear codes.

Our proof is simple, combining the Zyablov--Pinsker \cite{zyablov1981list} argument with recent analyses of the list-recovery of explicit constructions like Folded Reed--Solomon codes.
In particular, the Zyablov--Pinsker argument \cite{zyablov1981list} shows that a random linear code can be list-recovered so that, with high probability the output list always lies in a subspace of dimension at most $O(\ell/\varepsilon)$.
Naively, this implies an output list size bound of $q^{O(\ell/\varepsilon)}$.
However, recent analyses of list-recovering explicit codes \cite{kopparty2018improved,tamo2023tighter} showed that subspaces of dimension $D$ with good distance --- random linear codes are well known to have good distance with high probability --- can have at most $(\ell/\varepsilon)^{O(D)}$ points inside an $\ell$-list-recovery ball, thus giving our improved output list size.
We also show that we get the best possible output list size for our proof technique, in the sense that, for any linear code, there are output lists that span a subspace of dimension at least  $\Omega(\ell/\varepsilon)$ (see Proposition~\ref{pr:lb}).

\paragraph{List recovery for Random Reed-Solomon Codes.} 
Reed--Solomon codes \cite{RS60} are the most fundamental evaluation codes. A Reed--Solomon code is given by $n$ evaluation points $\alpha_1, \alpha_2, \ldots, \alpha_n$ in a finite field $\mathbb{F}_q$, and a degree $k<n$, and is defined as
\[
\text{RS}_{n,k}(\alpha_1, \ldots, \alpha_n) := \left \lbrace \left( f(\alpha_1), \ldots, f(\alpha_n) \right) \mid f\in \mathbb{F}_q[x],\text{ }\deg f < k \right \rbrace.
\]	

List-decoding and list-recovery of Reed--Solomon codes are well-studied questions.
The seminal Guruswami--Sudan \cite{guruswami1999improved} algorithm showed that Reed--Solomon codes are list-decodable and list-recoverable up to the \emph{Johnson radius} $1-\sqrt{R\ell}$ \cite{johnson1962new, guruswami2001extensions}.
Since then, there has been much interest in determining whether Reed--Solomon codes are list-decodable and list-recoverable beyond the Johnson bound, and perhaps even up to capacity $\rho=1-R$ (the capacity is $1-R$ for both list-decoding and list-recovery).
Initially, there was evidence against this possibility \cite{guruswami2006limits,cheng2007list,ben-sasson2010subspace}, suggesting that Reed--Solomon codes could not be list-decoded or list-recovered much beyond the Johnson bound.
Since then, an exciting line of work has shown, to contrary, that Reed--Solomon codes can beat the Johnson bound for list-decoding \cite{rudra2014every,shangguan2020combinatorial, ferber2022listdecodability,goldberg2022singletontype,guo2022improved, brakensiek2023generic, guo2023randomly, AGL24}, and, in fact, can be list-decoded up to capacity \cite{brakensiek2023generic, guo2023randomly, AGL24}. All of these works studied \emph{randomly punctured} Reed--Solomon codes, where $\alpha_1,\dots,\alpha_n$ are chosen at random from a larger field $q$.

Despite the exciting progress for list-decoding, there has been comparatively little progress on list-recovery.
Lund and Potukuchi \cite{lund2020list} and Guo, Li, Shanggaun, Tamo, and Wootters \cite{guo2022improved} proved that (randomly punctured) Reed--Solomon codes are list-recoverable beyond the Johnson bound in the low-rate regime: \cite{lund2020list} shows $(\rho,\ell,L)$-list-recovery for $\rho\le 1-1/\sqrt{2}$, $L=O(\ell)$ and rate $O\inparen{\frac{1}{\sqrt{\ell}\log q}}$, and \cite{guo2022improved} shows $\inparen{\Omega\inparen{\frac{\varepsilon}{\sqrt{\ell}\log(1/\varepsilon)}},\ell, O(\ell/\varepsilon) }$-list-recovery for rate $1-\varepsilon$ Reed--Solomon codes.
Both improve on the Johnson radius of $O\inparen{\frac{1}{\ell} }$ in the low rate setting.

In \cite{levi2024random}, Levi, Mosheiff and Shagrithaya showed that random Reed--Solomon codes and random linear codes are \emph{locally equivalent}, meaning that both random code families achieve identical rate thresholds for all ``local properties," which include (the complements of) list-decoding and list-recovery.
Thus, our result for list-recovery of random linear codes transfers to random RS codes as well. 
\begin{theorem}[Theorem~\ref{cor:list-rec-rs}, Informal]
    For all $R,\varepsilon\in(0,1)$, a randomly punctured Reed--Solomon code of length $n$ over alphabet size $q=n\cdot (\ell/\varepsilon)^{(\ell/\varepsilon)^{O(\ell/\varepsilon)}}$, of rate $R$ is $\inparen{1-R-\eps, \ell, \inparen{\frac{\ell}{\eps}}^{O(\ell/\eps)}}$-list-recoverable with high probability.
    \label{thm:main-2}
\end{theorem}

\begin{remark}
    We note that in order to use the equivalence result from \cite{levi2024random}, it is crucial that the upper bound on the list size be independent of the alphabet size, as guaranteed by Theorem~\ref{thm:main-1}. Hence, previous results on list size cannot be used with the equivalence result.
    \label{rem:main-2}
\end{remark}

\begin{remark}
    A fruitful line of work \cite{guruswami2008explicit,kopparty2015listdecoding,guruswami2013linearalgebraic,dvir2012subspace,kopparty2018improved,tamo2023tighter} has culminated in output list sizes of $O\inparen{\frac{\ell}{\eps}}^{O(\log(\ell)/\varepsilon)}$ for various explicit list-recoverable codes such as Folded Reed--Solomon codes and Multiplicity codes. 
    This list size is better than our list size of $(\frac{\ell}{\varepsilon})^{O(\ell/\varepsilon)}$ by roughly a factor of $\ell/\log(\ell)$ in the exponent.
    However, our results are still interesting because, as described above, the list-recovery of random linear codes and Reed--Solomon codes are fundamental questions, and also because our results yield linear codes for list-recovery and use much smaller alphabet sizes.
\end{remark}

\paragraph{Lower bounds for list-recovery.}
We now discuss impossibility results for list-recovery.
An early impossibility result of Guruswami and Rudra in \cite{guruswami2006limits} showed that, in the setting of zero-error list-recovery ($\rho=0$), many full length Reed--Solomon codes of rate $R$ require $R\le 1/\ell$ in order to have $\poly n$ output list size, so many full length Reed--Solomon codes cannot be list-recovered beyond the Johnson bound --- note this does not contradict Theorem~\ref{thm:main-2}, as we consider randomly punctured, as opposed to full length ($n=q$) codes.
More recently it was shown that achieving list-recovery capacity requires exponential list size $\ell^{\Omega(1/\varepsilon)}$ for particular codes: random linear codes in the high-rate zero-error ($\rho=0$) regime \cite{guruswami2021bounds}, random linear codes in general parameter settings \cite{levi2024random}, and for Reed--Solomon codes, Folded Reed--Solomon codes, and Multiplicity codes in general parameter settings \cite{chen2024explicit}. 

Inspired by the lower bound in \cite{chen2024explicit}, we show that \emph{any} linear code list-recoverable to capacity must have output list size at least $\ell^{\Omega(R/\varepsilon)}$. 
\begin{theorem}[Theorem~\ref{thm:lb-1}, Informal]
Over any field, any linear code of rate $R$ that is $(1-R-\varepsilon, \ell,L)$ list-recoverable must satisfy $L\ge \ell^{\Omega(R/\varepsilon)}$.
\label{thm:lb}
\end{theorem}

One takeaway from Theorem~\ref{thm:lb} is that our list sizes of $(\ell/\varepsilon)^{O(\ell/\varepsilon)}$ in Theorem~\ref{thm:main-1} and Theorem~\ref{thm:main-2} are near-optimal.
Additionally, Doron and Wootters \cite{doron2020highprobability} asked whether there were explicit list-recoverable codes with, among other desired guarantees, output list size $L=O(\ell)$. Our result shows this is not possible with \emph{any} linear code.
Lastly, our lower bound shows separation between non-linear and linear codes for list-recovery, which is perhaps surprising given that no such separation exists for list-decoding.

We point out that, for list-decoding ($\ell=1$), our lower bound is trivial ($L\ge 1$), so it does not contradict the recent results that random linear codes, randomly punctured Reed--Solomon codes, and randomly punctured Algebraic-Geometry codes achieve list-decoding capacity with output list size $O(1/\varepsilon)$ \cite{brakensiek2023generic, guo2023randomly, AGL24, brakensiek2023ag}.

\section{Preliminaries}
For a prime power $q$, let $\F_q$ be the finite field of order $q$.
Let $[n]$ denote the set $\inset{1,\ldots,n}$.
For a given vector space $V$, let $\calL(V)$ denote the set of all subspaces of $V$.
For a given set $S$, let $2^S$ denote the power set of $S$.
For a vector $v$, let $v[i]$ denote its $i$th entry.

A code $\calC \subseteq \F_q^n$ is said to be linear if it is a linear subspace, and said to have rate $R \in (0, 1)$ if $R=dim(\calC)/n$.
We say $\calC$ has relative distance $\delta \in (0, 1)$ if $\forall c \in \calC, wt(c) \geq \delta \cdot n$, where $wt(c)$ denotes the number of non-zero entries in the codeword $c$.
A matrix $\mathbf{G} \in \F_q^{n \times Rn}$ containing linearly independent columns is said to be the \deffont{generator matrix} of $\calC$ if every codeword $c \in \calC$ can be constructed using some linear combinations of the columns in $\mathbf{G}$.
$\calC$ is said to \deffont{contain} a set of vectors $s_1,\ldots s_b \in \F_q^n$ if $s_i \in \calC$ for every $i \in [b]$. 

For a vector $x \in \F_q^n$ and sets $S_1,\ldots,S_n \subseteq \F_q$, the \deffont{agreement set} $\agr(x, S_1,\ldots,S_n)$ is defined as:
\[
    \agr(x, S_1,\ldots,S_n) \coloneqq \inset{i \in [n] \mid x[i] \in S_i}.
\]
A \deffont{$\rho$-radius $\ell$-list-recovery ball} $B(\rho, S_1\times\cdots\times S_n)$ is given by input lists $S_1,\dots,S_n\subseteq\mathbb{F}_q$ of size $\ell$, and is defined to be
\begin{align}
    B(\rho, S_1\times\cdots\times S_n) = 
    \inset{x\in\mathbb{F}_q^n: \agr(x, S_1\times \dots\times S_n)\ge (1-\rho)n}.
\end{align}
We can alternatively define $(\rho, \ell, L)$-list-recovery using the above definition: a code $\calC \subseteq \F_q^n$ is $(\rho, \ell, L)$-list-recoverable if every $\rho$-radius $\ell$-list-recovery ball $B$ contains at most $L$ codewords.

For $0<R<1$, a \emph{random linear code} (RLC) of rate $R$ is a linear code whose generator matrix $\mathbf{G} \in \F_q^{n \times Rn}$ is a matrix whose entries are chosen uniformly at random from $\F_q$, independently of one another.
For $\alpha_1,\ldots,\alpha_n \in \F_q$, we use $\CRS{\alpha_1,\ldots,\alpha_n}{Rn}$ to denote the Reed--Solomon (RS) code of rate $R$ obtained by evaluating polynomials of degree $<Rn$ on evaluation points $\alpha_1,\ldots,\alpha_n \in \F_q$.
We say this is a \deffont{random RS code} if the evaluation points have been chosen uniformly at random and independently of one another\footnote{This is different from the usual model for random RS codes, where it is required that the random evaluation points be distinct. However, it can be shown that both models behave similarly (refer to \cite{levi2024random}, Appendix A for details).}.

\subsection{Local Coordinate-Wise Linear (LCL) Properties}
We now introduce the machinery in \cite{levi2024random} that connects random linear codes to (randomly punctured) Reed--Solomon codes.
A \deffont{code property} $\calP_n$ for codes of block length $n$ in $\F_q^n$ is simply a family of codes in $\F_q^n$. We say that a code $\calC_n \subseteq \F_q^n$ \deffont{satisfies} $\calP_n$ if $\calC_n \in \calP_n$.
Denoting $\calP \coloneqq \inset{\calP_n}_{n \in \NN}$, we say that an infinite family of codes $\calC_n \coloneqq \inset{\calC_n}_{n \in \NN}$ \deffont{satisfies} $\calP$ if $\calC_n \in \calP_n$ for every $n \in \NN$.
In this paper, we focus on local, monotone-increasing code properties.
A local code property, informally speaking, is defined by the inclusion of some bad set. A \deffont{monotone-increasing} code property is one for which the following is true: if $\calC$ satisfies $\calP$, then every $\calC'$ for which $\calC' \supseteq \calC$ holds, also satisfies $\calP$.
An example of local, monotone-increasing code property is the complement of $(\rho, L)$-list-decodability, defined as the family of all codes that contain at least one set of $L+1$ distinct vectors, all lying within a Hamming ball of radius $\rho$.

For a locality parameter $b \in \NN$, an ordered tuple of subspaces $\calV= (\calV_1, \ldots, \calV_n)$, where $\calV_i \in \F_q^b$ for each $i \in [n]$ is defined to be a \deffont{$b$-local profile}.
Note that $\calV \in \calL(\F_q^b)^n$.
We say that a matrix $A \in \F_q^{n \times b}$ \deffont{is contained in} $\calV$ if the $i$th row of A belongs to $\calV_i$, for all $i$.
A code $\calC \subseteq \F_q^n$ is said to \deffont{contain} $\calV$ if there exists a matrix $A \in \F_q^{n \times b}$ \emph{with distinct columns} such that the set of columns of $A$ is contained in $\calC$, and moreover, $A$ is contained in $\calV$.
For a family of $b$-local profiles $\inset{\calF_n}_{n \in \NN}$, where $\calF_n \subseteq \calL\left(\F_q^b\right)^n$, we define a \deffont{$b$-local coordinate wise linear} ($b$-LCL) property $\calP \coloneqq \inset{\calP_n}_{n \in \NN}$ as follows:
\[
    \calP_n = \inset{\calC \in \F_q^n \mid \exists \calV \in \calF_n \textrm{ such that } \calC \textrm{ contains } \calV}.
\]

The complement of $(\rho, \ell, L)$-list-recoverability is a $(L+1)$-LCL property. This is proven in \cite[Proposition 2.2]{levi2024random}, but we provide a justification in this paragraph.
Every bad set of vectors lying within a given $\rho$-radius $\ell$-list-recovery ball agrees with some input lists $S_1,\ldots,S_n \subseteq \F_q$ at a lot of coordinates.
This implies that the vectors agree with one another at a lot of coordinates as well, and once we arrange the bad vector sets as rows in a matrix of dimension $n \times (L+1)$, we can specify these agreements as linear constraints on the rows of such matrices.
Formally, the property is defined by a family of $(L+1)$-local profiles that we now describe.
For every $n \in \NN$, we define $\calF_n$ by describing the $(L+1)$-local profiles $\calV \in \calL(\F_q^{L+1})^n$ that constitute it.
Let $S_{(1-\rho)} \subseteq \left(2^{[n]}\right)^{L+1}$ denote the collection of all $L+1$-length tuples where each element is a subset of $[n]$ of size exactly $(1-\rho)n$.
Furthermore, let $M_{[\ell]} \coloneqq [\ell]^{n \times (L+1)}$ denote the set of all matrices of dimension $n \times (L+1)$ having elements in $[\ell]$ \footnote{Even though $S_{(1-\rho)}$ and $M_{[\ell]}$ depend on $n$, we have suppressed this dependence in the notation for sake of clarity.}.
Then, for every $s = (s_1,\ldots,s_{(L+1)}) \in S_{(1-\rho)}$, every $M \in M_{[\ell]}$, define $\calV(s, M) = (\calV_1,\ldots,\calV_n)$ such that for every $i \in [n]$,
\[
    \calV_i \coloneqq \inset{r \in \F_q^{L+1} \mid \forall j,k \in [L+1], r[j]=r[k] \textrm{ if } (i \in s_j) \land (i \in s_k) \land (M[i,j]=M[i,k])}.
\]
Note that each $\calV_i$ is a subspace of $\F_q^{(L+1)}$, and therefore $\calV(s, M)$ is a valid linear profile.
We can now define the associated family of linear profiles for the complement of $(\rho, \ell, L)$-list-recoverability:
\[
    \calF_n \coloneqq \inset{\calV \in \calL\left(\F_q^{L+1}\right)^n \mid \exists s \in S_{(1-\rho)}, M \in M_{[\ell]},~\textrm{ such that }\calV = \calV(s, M)}.
\]

Observe that
\begin{equation}\label{eq:up-bound-fam}
    \inabs{\calF_\calP} \le \inabs{S_{(1-\rho)}} \cdot \inabs{M_{[\ell]}} \le \binom{n}{\rho n}^{(L+1)} \cdot \ell^{(L+1)n}.
\end{equation}

In the same work, the authors also prove a threshold theorem for random linear codes (RLCs) in relation to all LCL properties, and moreover, gave a complete characterization of the rate threshold. Informally, the theorem says that RLCs of a sufficiently large alphabet exhibit a sharp threshold phenomenon for all LCL properties, that is, for every LCL property $\calP$, there exists a rate threshold $R_\calP$ such that RLCs of rate $R_\calP-\eps$ satisfy $\calP$ with high probability, and RLCs of rate $R_\calP+\eps$ \textbf{do not} satisfy $\calP$ with high probability.

\begin{theorem}[\cite{levi2024random}, Theorem 3.1]\label{thm:LMS-RLC-thresh}
    Let $\calP$ be a $b$-LCL property of codes in $\F_q^n$ and let $\calF\subseteq \calL\inparen{\F_q^b}^n$ be a corresponding family of profiles. Let $\calC\subseteq \F_q^n$ be an RLC of rate R. Then, there is some threshold rate $R_\calP$ for which the following holds.
    \begin{enumerate}
        \item If $R \ge R_{\calP}+\eps$ then $\Pr[\calC\textrm{ satisfies }\calP] \ge 1 - q^{-\eps n+b^2}$.
        \item If $R \le R_{\calP}-\eps$ then $\Pr[\calC\textrm{ satisfies }\calP] \le \inabs{\calF}\cdot q^{-\eps n + b^2}$.
        \item In particular, if $R \le R_{\calP}-\eps$ and $q \ge 2^{\frac{2\log_2 |\calF|}{\eps n}}$ then $\Pr[\calC\textrm{ satisfies }\calP] \le q^{-\frac{\eps n}2 + b^2}$.
    \end{enumerate}
\end{theorem}

The concept of LCL properties allows for ``transfer type" theorems between random linear codes and random RS codes. In more detail, for every reasonable LCL property (that is, for every LCL property whose corresponding family of profiles is large), the rate thresholds for random linear codes and random RS codes are the same. That is, any rate threshold proved for LCL properties of RLCs also applies for random RS codes, and vice versa. For our purposes, we only require one part of this result, which we formally state below.
\begin{theorem}[\cite{levi2024random}, Theorem~3.10 (part 1) (Threshold theorem for RS codes)]\label{thm:lms-red}
    Let $\calP$ be a $b$-LCL property of codes in $\F_q^n$, with associated local profile family $\calF_\calP \subseteq \calL\inparen{\F_q^b}^n$ and (random linear code) threshold rate $R_\calP$.
    Let $0 < R' < 1$ and let $\calC = \CRS[\F_q]{\alpha_1,\dots,\alpha_n}{R'n}$, and $\alpha_1,\dots,\alpha_n$ are sampled independently and uniformly from $\F_q$.
    Assume that $q > R'nb$.
    Fix $\eps' n \ge 2b(b+1)$.
    If $R'\le R_\calP-\eps'$, then 
          \begin{equation}\label{eq:MainReductionRLCtoRS}
              \Pr [\calC \textrm{ satisfies }\calP] \le (2^b-1) \cdot \inparen{\frac {(4b)^{4b}R'n}  {\eps' q}}^{\frac{\eps' n}{2b}} \cdot |\calF_\calP|.
          \end{equation}
\end{theorem}

\section{List-recovery of Random Linear Codes}\label{sec:list-rec-rlcs}
In this section, we prove Theorem~\ref{thm:main-1}, that random linear codes achieve list-recovery capacity with constant output list size. Formally, we show the following.
\begin{theorem}
    Fix $0 < R < 1$, $\eps>0$ so that $(1-R-\eps)>0$, $\ell \in \NN$, and let $q$ be a prime power such that $q \ge \max \inparen{\ell^{\frac{8R}{\eps}+6},\ell \cdot 2^{4/\eps}}$.
    Let $\calC \subseteq \F_q^n$ be an RLC of rate $R$.
    Then with probability at least $1-2q^{-\frac{\eps n}{8}}$, $\calC$ is $\inparen{(1-R-\eps), \ell, L}$-list-recoverable with $L$ satisfying $L \le \left(\frac{2\ell}{\eps}\right)^{2\ell/\eps}$.
    \label{thm:list-rec-rlcs}
\end{theorem}

The theorem follows as a consequence of two lemmas. We first state both lemmas, and then give the proof of Theorem~\ref{thm:list-rec-rlcs} using them. The first lemma essentially states that any low dimensional subspace with good distance has few points in a list-recovery ball. This lemma appears in \cite{kopparty2018improved,tamo2023tighter}; we state the version from \cite[Lemma 3.1]{tamo2023tighter}.

\begin{lemma}[\cite{tamo2023tighter}, see also \cite{kopparty2018improved}]
For $\eps>0$ and $\ell \in \NN$, let $\calC \subseteq \F_q^n$ be a linear code with relative distance $\delta > \eps/2$ that is $(\delta - \frac{\eps}{2}, \ell, L)$-list-recoverable. Assume further that any output list is contained in a subspace $V \subseteq \calC$ of dimension $r$.
Then the output list size $L \le \left(\frac{2\ell}{\eps}\right)^r$.
    \label{lem:tamo23}
\end{lemma}

The second lemma uses the Zyablov--Pinsker argument \cite{zyablov1981list}, showing that a random linear code does not have too many linearly independent codewords within a list-recovery ball.

\begin{lemma}
    Fix $0 < R < 1$, $\eps>0$ so that $(1-R-\eps)>0$, $\ell \in \NN$, and let $q$ be a prime power such that $q \ge \max \inparen{\ell^{\frac{8R}{\eps}+6},\ell \cdot 2^{4/\eps}}$.
    Let $\calC \subseteq \F_q^n$ be an RLC of rate $R$.
    Then with probability at least $1-q^{-\frac{\eps nL}{8}}$, for every input lists $\calS_1, \ldots, \calS_n$ of size $\ell$, the maximal linearly independent subset of $\calC$ within the $(1-R-\eps)$ radius $\ell$-list-recovery ball $B\inparen{ (1-R-\eps), S_1 \times  \cdots \times S_n}$ has size less than $2\ell/\eps$.
    \label{lem:zp}
\end{lemma}

\begin{proof}[Proof of Lemma~\ref{lem:zp}]
    Denote $\rho := 1-R-\eps$ and $L := 2\ell/\eps$.
    We also assume $q$ is a multiple of $\ell$ for simplicity of exposition, and note that the result holds in the general case as well.
    We show that an RLC ``avoids" all bad configurations with high probability.
    A \deffont{bad configuration} is a set $V$ of linearly independent vectors of size $L$ such that there exist input lists lists $\calS_1, \ldots, \calS_n$ of size $\ell$, so that $V \subseteq B\inparen{ \rho, S_1 \times  \cdots \times S_n}$.
    We say that $\calC$ contains a bad configuration $V$ if for every $v \in V$, $v$ is also in $\calC$. If this condition is not satisfied, then we say that $\calC$ does not contain $V$.
    It is easy to see that if $\calC$ contains no bad configurations, then the maximal linearly independent subset of $\calC$ within any $\ell$-list-recovery ball of radius $\rho$ has size less than $2\ell/\eps$.
    Therefore we show that the probability of $\calC$ containing a bad configuration is low.

    Fix input lists $\calS_1, \ldots, \calS_n$ of size $\ell$, and let $B := B\inparen{ \rho, S_1 \times  \cdots \times S_n}$ be the corresponding $\rho$ radius $\ell$-list-recovery ball. The size of $B$ is $\binom{n}{\rho n}\cdot \ell^{(1-\rho)n}\cdot (q-\ell)^{\rho n}$.
    The probability that a particular configuration is bad is equal to the probability of the encodings of some $L$ linearly independent messages being inside $B$ simultaneously, which is $\inparen{\frac{|B|}{q^n}}^L$.
    By a union bound over at most $q^{n\ell}$ possible input lists and all $L$-sized linearly independent subsets of the message vectors in $\F_q^{Rn}$ (there are at most $q^{RnL}$ such subsets), we have
    \begin{align}
        \Pr_\calC[\calC \text{ contains a bad configuration}]
        &\le \left(\frac{\binom{n}{\rho n}\cdot \ell^{(1-\rho )n}\cdot (q-\ell)^{\rho n}}{q^n}\right)^L \cdot q^{n\ell}\cdot q^{RnL} \nonumber\\
        &= \left(\left(\frac{\ell}{q}\right)^{n}\cdot \binom{n}{\rho n}\cdot \left(\frac{q}{\ell}-1\right)^{\rho n}\right)^L \cdot q^{n\ell} \cdot q^{RnL}\nonumber\\
        &\le \left(\left(\frac{\ell}{q}\right)^{n}\cdot (q/\ell)^{H_{q/\ell}(\rho )n}\right)^L \cdot q^{n\ell}\cdot q^{RnL} \nonumber\\       
        &\le \left((q/\ell)^{-(1-H_{q/\ell}(\rho ))n}\right)^L \cdot q^{n\ell}\cdot q^{RnL} \nonumber\\       
        &\le \left((q/\ell)^{-\inparen{R+\frac{3\eps}{4}}n}\right)^L \cdot q^{n\ell}\cdot q^{RnL} \label{eqn:entropy-req}\\ 
        &= \ell^{(R+\frac{3\eps}{4})nL}\cdot q^{-\inparen{R+\frac{3\eps}{4}}nL}\cdot q^{\frac{\eps nL}{2}}\cdot q^{RnL}\nonumber\\
        &= \ell^{(R+\frac{3\eps}{4})nL}\cdot q^{-\frac{\eps nL}{4}} \nonumber\\
        &\leq q^{-\frac{\eps nL}{8}}. \nonumber
    \end{align}
    In Equation~\ref{eqn:entropy-req}, we used $q \geq \ell \cdot 2^{4/\eps}$, and for the last inequality, we used $q \geq \ell^{\frac{8R}{\eps}+6}$. This implies that the probability with which $\calC$ does not contain any bad configuration is at least $1-q^{-\frac{\eps nL}{8}}$.
\end{proof}

We now prove Theorem~\ref{thm:list-rec-rlcs}.
\begin{proof}[Proof of Theorem~\ref{thm:list-rec-rlcs}]
    Denote $\rho := 1-R-\eps$.
    Denote by $E_1$ the event that for a RLC $\calC$ of rate $R$, the maximal linearly independent subset of $\calC$ within every $(1-R-\eps)$ radius $\ell$-list-recovery ball has size less than $2\ell/\eps$. 
    By Lemma~\ref{lem:zp}, we know that $E_1$ happens with probability at least $1-q^{-\frac{\eps nL}{8}}$.
    Let $E_2$ denote the event that a rate $R$ RLC $\calC$ has distance at least $1-R-\frac{\eps}{2}$.
    By the Gilbert-Varshamov bound (see \cite{guruswami2022essential}, Section 4.2), and because of the fact that $q\ge \ell \cdot 2^{4/\eps} \ge 2^{4/\eps}$, $E_2$ happens with probability at least $1-q^{-\frac{\eps n}{2}}$.
    Therefore we have
    \[
        \Pr_\calC[E_1 \land E_2] \ge 1-q^{-\frac{\eps nL}{8}}-q^{-\frac{\eps n}{2}} \ge 1-2\cdot q^{-\frac{\eps n}{8}}
    \]
    When $E_1$ and $E_2$ occur simultaneously, the assumptions of Lemma~\ref{lem:tamo23} are satisfied by $\calC$ with $r=2\ell / \eps$ and $\delta=1-R-\frac{\eps}{2}$, and therefore we see that
    \[
        \Pr_\calC\inbrak{\bigwedge_B \inabs{ \calC \cap B} \le \left(\frac{2\ell}{\eps}\right)^{2\ell / \eps}} \ge 1-2\cdot q^{-\frac{\eps n}{8}}
    \]
    where $B$ is ranging over all $(1-R-\eps)$ radius $\ell$-list-recovery balls, and we are done.
\end{proof}

\section{List-Recovery of Reed--Solomon codes}\label{sec:list-rec-rs}
In this section, we will prove the following result, which says that random Reed--Solomon codes are list-recoverable to capacity with constant output list size. The proof combines Theorem~\ref{thm:list-rec-rlcs} from the previous section with Theorem~\ref{thm:LMS-RLC-thresh}, the equivalence theorem from \cite{levi2024random}.
\begin{corollary}\label{cor:list-rec-rs}
    Fix $0 < R < 1$, $\eps>0$ so that $(1-R-\eps)>0$, $\ell \in \NN$.
    Fix a constant $\eps'>0$ such that $\eps' < R$, and denote $L \coloneqq \floor{\left(\frac{2\ell}{\eps}\right)^{2\ell/\eps}}$.
    Let $\eta>0$ be a constant, and let $q$ be a prime power satisfying $q > \frac{(4(L+1))^{4(L+1)}Rn}{\eps'}\cdot 2^{\frac{((\log \ell + 2)(L+1)+ \eta)\cdot 2(L+1)}{\eps'}}$.
    Then, a random RS code of rate $R-\eps'$ over $\F_q^n$ is $(1-R-\eps, \ell, L)$-list-recoverable with probability at least $1-2^{-\eta n}$.
\end{corollary}

\begin{proof}[Proof of Corollary~\ref{cor:list-rec-rs}]
    Denote $L \coloneqq \floor{\left(\frac{2\ell}{\eps}\right)^{2\ell/\eps}}$ and $b \coloneqq L+1$. Let $\calP$ be the $b$-LCL property of \textbf{not} being $(1-R-\eps, \ell, L)$-list-recoverable, and let $R_{\calP}$ be the corresponding (random linear code) threshold rate.
    By Theorem~\ref{thm:LMS-RLC-thresh}, part 1 \cite{levi2024random}, we know that if $\calC$ is an RLC of rate $R$, then the following holds for every constant $\eps^*>0$:
    \[
        \Pr [\calC \textrm{ satisfies }\calP] < 1-q^{\eps^*n + b^2} \implies R < R_\calP+ \eps^*
    \]
    According to Theorem~\ref{thm:list-rec-rlcs}, a rate $R$ RLC (having a sufficiently large alphabet size) satisfies $\calP$ only with probability at most $2q^{-\frac{\eps n}{8}} < 1-q^{\eps^*n + b^2}$.
    Therefore, $R < R_\calP+ \eps^*$ for every $\eps^*>0$, and so $R \le R_\calP$.

    We will now work with random RS codes having rate slightly less than $R$.
    Define $R' \coloneqq R-\eps' \le R_\calP-\eps'$, take $n$ to be large enough so that $\eps' n \ge 2b(b+1)$.
    Note that $q > Rnb > R'nb$.
    Define $\calC = \CRS[\F_q]{\alpha_1,\dots,\alpha_n}{R'n}$, where $\alpha_1,\dots,\alpha_n$ are sampled independently and uniformly from $\F_q$.
    Upon denoting $\calF_\calP$ to be the local profile family associated with property $\calP$, we see that the hypothesis of Theorem~\ref{thm:lms-red} \cite{levi2024random} is satisfied, and therefore, Equation~\ref{eq:MainReductionRLCtoRS} is satisfied.

    Recall that we calculated an upper bound for $\inabs{\calF_\calP}$ in Equation~\ref{eq:up-bound-fam}, and so $\inabs{\calF_\calP} \le \binom{n}{(1-R-\eps)n}^{b}\cdot \ell^{bn}$.
    Substituting this bound on $\inabs{\calF_\calP}$ in Equation~\ref{eq:MainReductionRLCtoRS},
    \begin{align*}
        \Pr [\calC \textrm{ satisfies }\calP] &\le (2^b-1) \cdot \inparen{\frac {(4b)^{4b}R'n}  {\eps' q}}^{\frac{\eps' n}{2b}} \cdot |\calF_\calP| \\
        &\le \inparen{\frac {(4b)^{4b}R'n}  {\eps' q}}^{\frac{\eps' n}{2b}} \cdot \binom{n}{(1-R-\eps)n}^{b}\cdot \ell^{bn} \\
        &\le \inparen{\frac {(4(L+1))^{4(L+1)}R'n}  {\eps' q}}^{\frac{\eps' n}{2(L+1)}} \cdot 2^{(H_2(1-R-\eps)+1)\cdot (L+1) n}.
    \end{align*}
    Because $q > \frac{(4(L+1))^{4(L+1)}R'n}{\eps'}\cdot 2^{\frac{((\log \ell + 2)(L+1)+ \eta)\cdot 2(L+1)}{\eps'}}$, we see that $\Pr [\calC \textrm{ satisfies }\calP] \le 2^{-\eta n}$.
    Thus, $\calC$ is $(1-R-\eps, \ell, L)$-list-recoverable with probability at least $1-2^{-\eta n}$.
\end{proof}

\section{Any linear code needs output list-size $\ell^{\Omega(R/\epsilon)}$}\label{sec:list-rec-impossibility}
We now prove our lower bounds for list-recovery, that any linear code list-recoverable to capacity needs output list size $\ell^{\Omega(R/\varepsilon)}$.
\begin{theorem}
    Let $R,\eps\in(0,1)$, $\ell$ be a positive integer, and $n\ge n_0(\ell,R,\eps)$ be sufficiently large.
    Let $\calC \subseteq \F^n$ be a linear code of rate $R$. If $\calC$ is  $(1-R-\eps, \ell, L)$-list-recoverable, then $L> \ell^{\floor{R/\eps}}$.
    \label{thm:lb-1}
\end{theorem}
\begin{proof}
    Let $k \coloneqq Rn$ be the dimension of the code.
    Let $k' = \left \lceil{\frac{\eps}{R}\cdot k} \right \rceil$.
    Let $m = \left \lfloor{\frac{k-1}{k'+1}} \right \rfloor$.
    By Gaussian elimination and permuting rows and columns, we may, without loss of generality write the generator matrix of $\calC$ as
    \begin{align}
        \mathbf{G} = 
        \begin{bmatrix}
            1 & & & & \\
              & 1 & & & \\
              &   & \ddots & &\\
              & & & \ddots &\\
              & & & & 1\\
              \hline
              * & * & \cdots & \cdots& * \\
        \end{bmatrix}
    \end{align}
    where each $*$ is a length $n-k$ vector.
    For $i \in [k]$, let $v_i\in\mathbb{F}^n$ denote the columns.
    By rank-nullity, there exist vectors $w_0,\dots,w_{m-1}\in\mathbb{F}^n$ such that $w_i$ is a linear combination of $v_{i\cdot (k'+1)+1},\dots,v_{(i+1)\cdot (k'+1)}$ such that $w_i$ is not supported on indices $k+1,\dots,k+k'$ (there are $k'+1$ vectors and $k'$ indices).
    Now let $w_{m} = v_k$.
    Restricted to indices in $[k+k']$, vectors $w_0,\dots,w_m$ have pairwise disjoint supports: within indices $[k+k']$, for $i=0,\dots,m-1$, vector $w_i$ is supported on $i\cdot (k'+1)+1,\dots, (i+1)\cdot(k'+1)\le k-1$, and vector $w_m$ is supported on $k,\dots,k+k'$.

    Now fix $\ell$ arbitrary distinct values $\beta_1,\dots,\beta_\ell\in\mathbb{F}_q$.
    Consider the output list $\calL = \{\sum_{i=0}^m \beta_{r_i}w_i: r_i\in[\ell]\}$ to be all linear combinations of $w_i$ with coefficients from $\beta_1,\dots,\beta_\ell$.
    The fact that the vectors $w_0,\dots,w_{m}$ have pairwise disjoint supports on $[k+k']$ implies (i) the vectors $w_0,\dots,w_{m}$ are linearly independent, and so all vectors in $\calL$ are distinct, and (ii) the vectors in $\mathcal{L}$ can only take on one of $\ell$ values at any index in $[k+k']$. Therefore, we can choose input lists $S_1,\dots,S_{k+k'}$, each of size $\ell$ such that all codewords in $\mathcal{L}$ agree with all of $S_1,\dots,S_{k+k'}$.

    Choosing the rest of the input lists arbitrarily, we have that this code is not $(\rho, \ell, L)$ list-recoverable with radius $\rho = (n-k-k')/n < 1-R-\eps$ and list size $L=\ell^{m+1} \ge \ell^{\floor{R/\eps}}.$\footnote{$m+1 = \left \lfloor{\frac{k+k'}{k'+1}} \right \rfloor \ge \left \lfloor{\frac{k+\frac{\eps}{R} k}{\frac{\eps}{R} k +2}}\right \rfloor \ge \left \lfloor{\frac{R}{\eps}} \right \rfloor$, where we used that $k>2R^2/\eps^2$.}

\end{proof}

We also show that our Zyablov-Pinsker type argument in Theorem~\ref{thm:main-1} (Lemma~\ref{lem:zp}) is tight, in the sense that any linear code must have $\Omega(\ell/\varepsilon)$ linearly independent codewords in a list-recovery ball.

\begin{proposition}
    Let $R,\eps\in(0,1)$, $\ell$ be a positive integer, and $n\ge n_0(\ell,R,\eps)$ be sufficiently large.
    Let $\calC$ be a linear code of rate $R$. Then there exists a $(1-R-\eps)$ radius $\ell$-list-recovery ball $B$ that contains at least $\ceil{(1-R)\ell/\eps}-1$ linearly independent elements of $\calC$.
    \label{pr:lb}
\end{proposition}
\begin{proof}
    Upon writing the generator matrix of $\calC$ in the same form as described above in the proof of Theorem~\ref{thm:lb-1}, consider the first $m \coloneqq \ceil{(1-R)\ell/\eps}-1 < (1-R)\ell/\eps$ columns of the generator matrix.
    Denote these linearly independent column vectors by $v_1,\ldots,v_m$.
    Create input lists $S_1,\dots,S_k$ each of size $\ell$ that contain $0$ and $1$, but are otherwise arbitrary.
    Then create lists $S_{k+1},\dots,S_n$ of size $\ell$, each containing elements that are evenly distributed so that they agree equally with each of $v_1,\dots,v_m$.
    Thus, each of $v_1,\dots,v_m$ agrees with $S_1,\dots,S_n$ on the first $k$, and on at least $\floor{\frac{\ell}{m} \cdot (n-k)} >\eps n$ of the remaining coordinates. Therefore, these vectors lie inside a $(1-R-\eps)n$-radius $\ell$-list-recovery ball around $S_1,\dots,S_n$, as desired.
\end{proof}

\section{Concluding remarks}

We showed that random linear codes and Reed--Solomon codes are list-recoverable to capacity with near-optimal output-list size. Several open questions remain.
\begin{enumerate}
    \item What is the optimal output-list size for random linear codes and Reed--Solomon codes? There is a gap between our upper bound of $(\frac{\ell}{\varepsilon})^{O(\ell/\varepsilon)}$ and the lower bound of $\ell^{\Omega(1/\varepsilon)}$. We surmise that the correct answer is closer to the lower bound.

    \item As asked by Doron and Wootters \cite{doron2020highprobability}, are there \emph{explicit} list-recoverable codes with output list size $L=O_\varepsilon(\ell)$? (and, even better, over alphabet size $q=\poly(\ell)$). We showed (Theorem~\ref{thm:lb}) that any such code must be nonlinear.
    
    \item Our alphabet size for list-recovering Reed--Solomon codes (Theorem~\ref{thm:main-2}) is optimal in that it is linear in $n$, but the constant is double-exponential in $\ell/\varepsilon$. By contrast, for list-decoding, the best known alphabet size for achieving capacity has an exponential-type constant, $2^{\poly (1/\varepsilon)} \cdot n$ \cite{AGL24}. Can our alphabet size be improved?
\end{enumerate}

\section{Acknowledgments}
The authors would like to thank Yeyuan Chen and Zihan Zhang for pointing out a mistake in Theorem~\ref{thm:lb-1} in an earlier version of the paper.
\clearpage
\printbibliography

\end{document}